\documentclass[12pt]{article}

\usepackage[utf8]{inputenc}
\usepackage[T1]{fontenc}
\usepackage{mathpazo}
\usepackage[final,expansion=false]{microtype}

\usepackage[letterpaper, margin=1in]{geometry}
\usepackage{setspace}
\usepackage{amsmath, amssymb, amsfonts, amsthm}
\usepackage{mathtools}
\usepackage{bm}

\theoremstyle{plain}

\newtheorem{proposition}{Proposition}

\theoremstyle{definition}

\newcommand{\one}{\mathbf{1}}

\newcommand{\R}{\mathbb{R}}
\newcommand{\Sc}{\mathcal{S}}
\newcommand{\Hc}{\mathcal{H}}
\newcommand{\clr}{\operatorname{clr}}
\newcommand{\Cl}{\mathcal{C}}

\usepackage{graphicx}
\usepackage[font={footnotesize},justification=justified]{caption}
\usepackage{booktabs}
\usepackage{adjustbox}
\usepackage{threeparttable}

\usepackage{enumitem}

\usepackage[svgnames]{xcolor}
\usepackage[colorlinks=true, linkcolor=black, citecolor=black,
            urlcolor=blue]{hyperref}
\usepackage{natbib}
\usepackage[capitalize, noabbrev, nameinlink]{cleveref}
\crefname{assumption}{Assumption}{Assumptions}

\usepackage{tikz}
\usetikzlibrary{calc, positioning}
\usepackage{pgfplots}
\pgfplotsset{compat=1.18}
\usepgfplotslibrary{ternary}

\title{\large\bf Compositional Synthetic Controls}
\author{Onil Boussim\thanks{Department of Economics, California Polytechnic State University. Email: \texttt{oboussim@calpoly.edu}. I thank Marc Henry and participants of the Penn State econometrics study group.}}
\date{}

\begin{document}
\maketitle
\vspace{-1.5em}

\begin{abstract}
\noindent

When applying synthetic control to compositional outcomes (budget, vote shares), researchers commonly minimize Euclidean distances between raw shares. I propose constructing synthetic controls in Aitchison geometry, using centered log-ratio coordinates to select donor weights and form counterfactuals. This approach respects proportional comparisons and is exactly invariant to common multiplicative changes in relative odds. Under a multinomial-choice model, its weights summarize similarity in relative utility indices. Monte Carlo simulations, including a CES allocation model, show that this method performs better when relative incentives determine shares. An application to U.S. school finance demonstrates that this choice can alter reported inference.

\end{abstract}

\section{Introduction}

Compositional outcomes are important in applied economics: school revenue by source, electricity generation by source, value added by sector, and vote shares. When evaluating policies that affect these outcomes using the synthetic control method \citep{abadie2003economic,abadie2010synthetic}, researchers commonly apply the method separately to each component or match raw shares using Euclidean distance. Recent work shows that imposing common donor weights across components ensures that estimated effects sum to zero and can improve efficiency by exploiting shared latent factors \citep{bogatyrev2026estimating,tian2023synthetic,sun2025using}.

However, this literature leaves open a central question: \textit{in which coordinates should the synthetic control be constructed?} Should researchers match donors on raw shares using Euclidean distance, or on log-ratios using the geometry of the simplex \citep{aitchison1982statistical}? This is not a tuning choice. The metric determines both the donor weights, which minimize pre-treatment discrepancies, and the counterfactual, which is the Fr\'echet mean under that same metric.

This paper shows that clr SCM is the theoretically more coherent and empirically preferred specification for intrinsically compositional outcomes. First, Euclidean fitting compares absolute share differences rather than proportional discrepancies; clr fitting instead respects the geometry of relative allocation. Second, clr SCM is invariant to common multiplicative perturbations. Such perturbations arise naturally when aggregate shocks shift relative odds, as in multinomial-choice models; raw-share SCM has no corresponding general invariance. Under that model, clr weights also summarize similarity in relative utility indices, giving donor selection a behavioral interpretation. Third, simulations and a school-finance application show that these theoretical differences have practical consequences. clr SCM performs better under a utility-based CES allocation model and under the clr-factor model implied by relative-choice behavior. Raw-share SCM performs better only when shares are generated by an additive raw-share factor process, which is generally not closed on the simplex, can imply negative shares without ad hoc restrictions, and lacks a natural foundation in models of relative choice or allocation. Thus, although an analyst may prefer raw-share loss, clr SCM is the appropriate default for coherent economic models of compositional outcomes.

The remainder of this paper proceeds as follows. Section \ref{sec:setup} formalizes the two methods. Section \ref{sec:why} details the theoretical advantages of the Aitchison approach. Section \ref{sec:sim} presents simulation evidence, and Section \ref{sec:app} applies the methods to U.S.\ school finance data. Section \ref{sec:concl} concludes.

\section{The two estimators}\label{sec:setup}
Unit $j\in\{1,\dots,J\}$ is observed for $t=1,\dots,T$ with composition $\pi_{j,t}$ in the
open simplex $\Sc^{p-1}=\{v\in\R^p_{>0}:\one'v=1\}$; unit $1$ is treated after $T_0$; and
$\Lambda=\{w\in\R^{J-1}:w_j\ge0,\ \sum_{j\ge2}w_j=1\}$. Write $\Cl(x)=x/\one'x$ and let
\[
\delta\oplus\pi=\Cl(\delta_1\pi_1,\dots,\delta_p\pi_p),\qquad \delta\in\R^p_{>0},
\]
denote perturbation, the group operation of the simplex. The map
$\clr(\pi)=\bigl(\log\pi_k-p^{-1}\sum_\ell\log\pi_\ell\bigr)_{k\le p}$ is a group isomorphism
onto $(\Hc,+)$ with $\Hc=\{z\in\R^p:\one'z=0\}$, satisfies
$\clr(\delta\oplus\pi)=\clr(\delta)+\clr(\pi)$, and induces the Aitchison distance
$\delta_A(\pi,\varpi)=\|\clr(\pi)-\clr(\varpi)\|_2$. A synthetic control selects weights by matching a pre-treatment path and then averages the donors. A distance determines both, since the average is its Fr\'echet mean
$\arg\min_m\sum_j w_j d^2(\pi_j,m)$.

Under Euclidean distance on shares the Fr\'echet mean is the weighted arithmetic mean
$M_w(\pi)=\sum_j w_j\pi_j$. Under $\delta_A$ it is the closed weighted geometric mean, or
Aitchison barycenter, $B_w(\pi)=\clr^{-1}\bigl(\sum_j w_j\clr(\pi_j)\bigr)$. The two estimators are therefore
\begin{align}
\hat w^{\mathrm{E}}&=\arg\min_{w\in\Lambda}\sum_{t\le T_0}
\Bigl\|\pi_{1,t}-\sum_{j\ge2}w_j\pi_{j,t}\Bigr\|_2^2,
\qquad \hat\pi^{0,\mathrm{E}}_{1,t}=M_{\hat w}(\pi_{\cdot,t}),
\label{eq:sh}\\[2pt]
\hat w^{\mathrm{A}}&=\arg\min_{w\in\Lambda}\sum_{t\le T_0}
\delta_A^2\bigl(\pi_{1,t},\,B_w(\pi_{\cdot,t})\bigr),
\qquad \hat\pi^{0,\mathrm{A}}_{1,t}=B_{\hat w}(\pi_{\cdot,t}).
\label{eq:clr}
\end{align}
Both deliver effects summing to zero, and \eqref{eq:clr} costs nothing to implement: it is
\eqref{eq:sh} run on the $p$ stacked clr series with fitted values mapped back by
$\clr^{-1}$.

\section{Why Aitchison geometry is appropriate for compositional outcomes}\label{sec:why}

\paragraph{Information lies in ratios.}
Shares lie on the simplex: their economic content is relative allocation, not independent arithmetic levels. Consider three regions allocating GDP across services, industry, and agriculture. Region A allocates $(82.5\%,15.0\%,2.5\%)$, Region B $(72.5\%,25.0\%,2.5\%)$, and Region C $(92.5\%,5.0\%,2.5\%)$. Euclidean distance places B and C at the same distance from A (about $0.14$). Yet the services-to-industry ratio is $5.5$ in A, $2.9$ in B, and $18.5$ in C. Since ratios describe the relevant economic structure, B should be the closer comparator. 

\paragraph{Latent factor models.} Raw-share SCM is naturally associated with an additive factor model,
$\pi_{k,j,t}=\alpha_{k,t}+\lambda_{k,t}'\mu_j+\epsilon_{k,j,t}$. This model is not closed on
the simplex; for many parameter draws it predicts negative shares, making it theoretically
incoherent for compositional data. By contrast, the Aitchison method corresponds to a factor
model in clr space: $\clr(\pi_{j,t})=d_t+\Lambda_t\mu_j+e_{j,t}$. Because the inverse clr
transformation maps the entire hyperplane $\Hc=\{z\in\R^p:\one'z=0\}$ back into the open
simplex, predicted compositions are strictly positive and sum to one by construction,
avoiding boundary violations. Negative fitted shares arise only when raw-share SCM is extended to allow extrapolation, intercepts, or unconstrained regression weights.

\paragraph{Behavioral meaning of clr weights.} Suppose a composition records aggregate choices among $p$ mutually exclusive sources and follows a multinomial-choice rule,
\[
\pi_{k,j,t}=\frac{\exp(V_{k,j,t})}{\sum_{\ell=1}^{p}\exp(V_{\ell,j,t})},
\]
where $V_{k,j,t}$ is the systematic utility or attractiveness index of source $k$ for unit $j$ at time $t$ \citep{mcfadden1974measurement}. Then $\clr(\pi_{j,t})=HV_{j,t}$, where $H=I-p^{-1}\one\one'$ removes the arbitrary common level of utility. Consequently, clr SCM selects donor weights that minimize discrepancies in \emph{relative} utility indices. The weights therefore have a preference-similarity interpretation under this model, not merely a statistical matching interpretation. For revenue shares, ``preference'' should be read broadly as the relative attractiveness of financing sources to the fiscal system, including political preferences, statutory rules, and source-specific fiscal costs. 

\paragraph{Common changes in relative odds.}
In a multinomial-logit interpretation, a common utility change shifts relative odds multiplicatively: $\pi_{j,t}=\delta_t\oplus\tilde\pi_{j,t}$. The following result makes the implication for SCM explicit.

\begin{proposition}\label{prop:cancel}
Let every unit in period $t$ be perturbed by the same $\delta_t\in\Sc^{p-1}$. Then clr SCM has identical pre-treatment objective values, minimizing weights, and clr treatment effects whether it is applied to $\{\pi_{j,t}\}$ or $\{\delta_t\oplus\pi_{j,t}\}$. Raw-share Euclidean SCM is not invariant to this transformation in general.
\end{proposition}

\noindent\textit{Proof.} Since $\clr(\delta_t\oplus\pi_{j,t})=\clr(\delta_t)+\clr(\pi_{j,t})$ and $\sum_jw_j=1$, the common term cancels from each residual and the post-treatment contrast. For Euclidean SCM, closure after componentwise multiplication is nonlinear and depends on $\pi_{j,t}$; it cannot generally be factored out of a weighted arithmetic average. $\square$

This result does not imply that every aggregate movement is irrelevant, or that every economic shock is multiplicative. It identifies a setting in which the log-ratio specification removes a common nuisance component exactly.

\section{Simulations}\label{sec:sim}

\paragraph{Design.} I use a transparent, symmetric Monte Carlo design. Each replication has ten units, three shares, sixteen pre-treatment and ten post-treatment periods. The treated unit's latent factors are a convex combination of four donors' factors; all units also receive idiosyncratic noise. I report three DGPs. The first is additive in clr coordinates and the second is additive in raw shares. The third is an economic allocation model that is not specified in terms of either estimator: each state allocates revenue across sources according to CES shares,
\[
 \pi_{k,j,t}=\frac{\exp\{\theta_{k,j}+(1-\sigma)\log p_{k,j,t}+u_{k,j,t}\}}
 {\sum_{\ell}\exp\{\theta_{\ell,j}+(1-\sigma)\log p_{\ell,j,t}+u_{\ell,j,t}\}},
\]
where $\theta_{k,j}$ represents a state-specific propensity to rely on source $k$, $p_{k,j,t}$ is its effective fiscal price or administrative cost, and $\sigma=1.6$. Source prices have common time variation, state-specific exposure, and transitory disturbances. Thus, the shares arise from an economic allocation rule, not from an imposed SCM loss function. A common, time-varying log-ratio perturbation is either absent or applied to every unit. Counterfactual errors are calculated against the untreated treated-unit composition using both Aitchison distance and share RMSE. 

\begin{table}[t]
\caption{Monte Carlo mean counterfactual error}
\label{tab:sim}
\centering
\begin{threeparttable}
\small
\begin{tabular}{lcccc}
\toprule
& \multicolumn{2}{c}{Aitchison error} & \multicolumn{2}{c}{Share RMSE}\\
\cmidrule(lr){2-3}\cmidrule(lr){4-5}
& Aitch.\ method & Eucl.\ method & Aitch.\ method & Eucl.\ method\\
\midrule
\multicolumn{5}{l}{\emph{A. clr-factor DGP}}\\
No common perturbation & \textbf{0.162} & 0.169 & \textbf{0.0311} & 0.0322\\
Common log-ratio perturbation & \textbf{0.159} & 0.170 & \textbf{0.0301} & 0.0319\\
\addlinespace
\multicolumn{5}{l}{\emph{B. Raw-share-factor DGP}}\\
No common perturbation & 0.0548 & \textbf{0.0483} & 0.00819 & \textbf{0.00732}\\
Common log-ratio perturbation & 0.0535 & \textbf{0.0498} & 0.00805 & \textbf{0.00755}\\
\addlinespace
\multicolumn{5}{l}{\emph{C. CES allocation DGP}}\\
No common perturbation & \textbf{0.151} & 0.155 & \textbf{0.0286} & 0.0292\\
Common log-ratio perturbation & \textbf{0.153} & 0.159 & \textbf{0.0286} & 0.0298\\
\bottomrule
\end{tabular}
\begin{tablenotes}[flushleft]\footnotesize
\item \textit{Note.} 200 replications per cell. Errors are measured against the true untreated composition in the ten post-treatment periods. Best entry in each pair is in bold. The exact invariance result in Proposition \ref{prop:cancel} concerns a given panel; small differences across the two shock rows reflect independent Monte Carlo draws.
\end{tablenotes}
\end{threeparttable}
\end{table}

\paragraph{Results.} Table \ref{tab:sim} supplies an important qualification. Under the clr-factor DGP, clr SCM has lower error under both scores; under the raw-share DGP, raw-share SCM performs better under both scores. The CES allocation model is the more economically interpretable comparison: clr SCM again has lower error under both scores, despite the DGP not being written as an SCM factor model. This does not establish universal dominance. The CES formula implies that relative shares respond to relative propensities and effective prices, so it naturally gives content to log-ratio comparisons. 

\section{Application: New York's Gap Elimination Adjustment }\label{sec:app}

The outcome is the composition of public K--12 revenue across federal, state, and local
sources for 50 states from 1992 to 2016, from the U.S.\ Census Bureau's Annual Survey of
School System Finances.

In the enacted budget for 2010--11, New York introduced the Gap Elimination Adjustment, a
formula-based reduction applied to each district's state school aid, which cut aid by
roughly \$1.4 billion in its first year and was not fully restored until 2017. I set $T_0=2009$, giving eighteen pre-treatment and seven post-treatment years, with all 49 remaining states as donors.

\begin{table}[t]
\caption{New York, Gap Elimination Adjustment}
\label{tab:ny}
\centering
\begin{threeparttable}
\footnotesize
\setlength{\tabcolsep}{3pt}
\begin{tabular}{lcc}
\toprule
& Aitchison method & Euclidean method\\
\midrule
Pre-treatment Aitchison RMSPE & \textbf{0.051} & 0.074\\
Pre-treatment share RMSPE & 0.0099 & \textbf{0.0084}\\
\quad Federal: RMSPE (pp) / mean $|\!\log$ ratio$|$ & 0.24 / 0.033 & 0.48 / 0.059\\
\quad State & 1.26 / 0.021 & 1.02 / 0.018\\
\quad Local & 1.14 / 0.016 & 0.91 / 0.014\\
\addlinespace
Donors with weight & 14 (max Virginia 0.27) & 11 (max Ohio 0.54)\\
Rank among 50 states & \textbf{2} & 7\\
Permutation $p$-value & \textbf{0.040} & 0.140\\
\addlinespace
\multicolumn{3}{l}{\emph{Share effects} (pp): Federal / State / Local}\\
\quad 2011 & $-0.77$ / $-2.44$ / $+3.21$ & $-1.42$ / $-3.44$ / $+4.86$\\
\quad 2013 & $-1.37$ / $-4.07$ / $+5.44$ & $-2.18$ / $-3.16$ / $+5.34$\\
\quad 2016 & $-1.37$ / $-2.99$ / $+4.37$ & $-2.10$ / $-1.41$ / $+3.52$\\
\bottomrule
\end{tabular}
\begin{tablenotes}[flushleft]\footnotesize
\item \textit{Note.} 50 states, eighteen pre-treatment and seven post-treatment years.
$p$-values are exact one-sided permutation values under the label-symmetric rule, so the
smallest attainable value is $1/50=0.02$. Effects sum to zero across categories by
construction.
\end{tablenotes}
\end{threeparttable}
\end{table}

\begin{figure}[t]
\centering
\includegraphics[width=\textwidth]{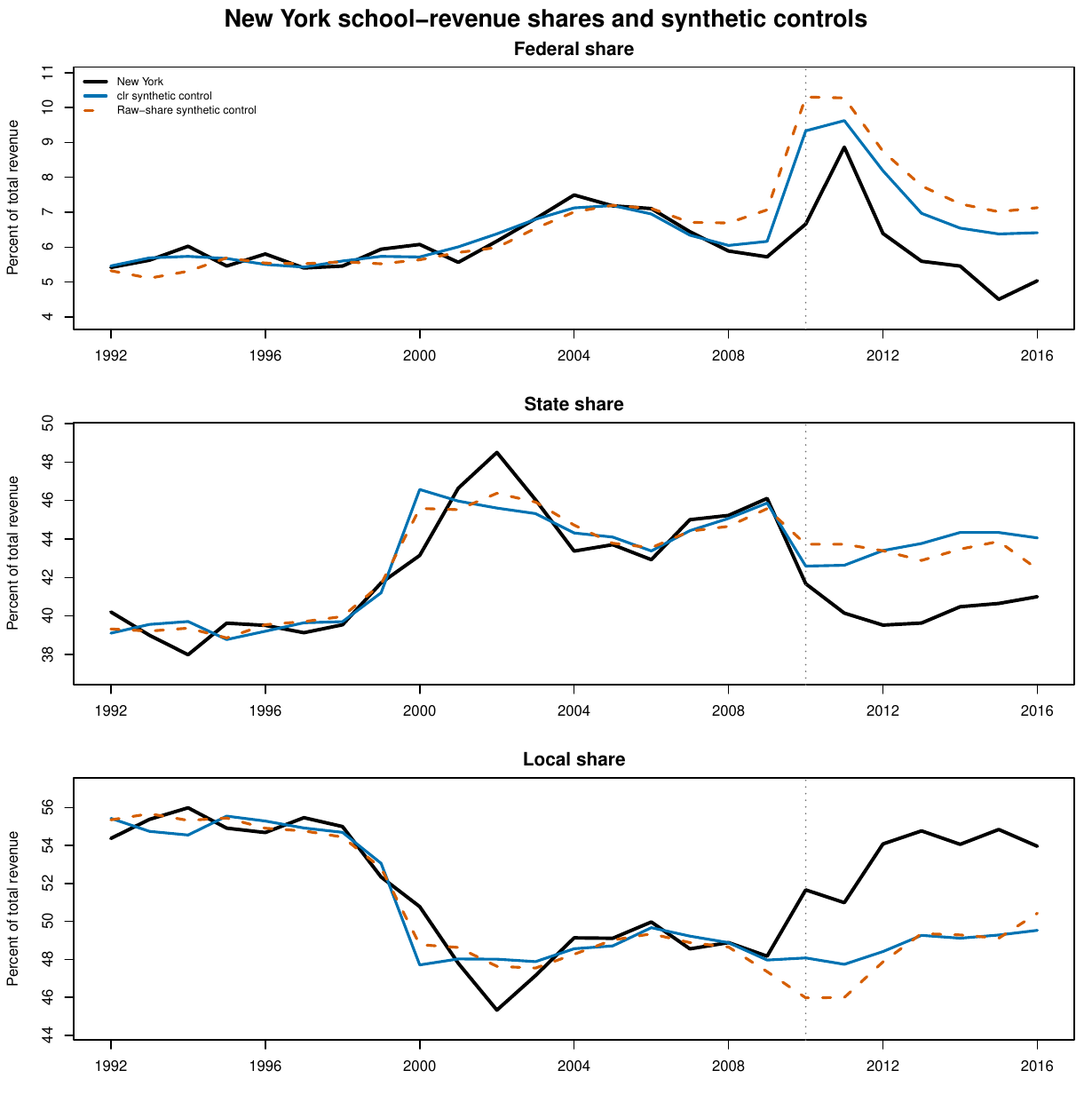}
\caption{New York's revenue shares and the two synthetic controls. The vertical line marks the 2010--11 Gap Elimination Adjustment. clr SCM is shown in blue and raw-share SCM in orange.}
\label{fig:nyshares}
\end{figure}

Both methods recover the sign the policy implies: the state share falls relative to its
synthetic counterpart and the local share rises by a corresponding amount, with the federal
share moving little (Table \ref{tab:ny}). Each fits the pre-treatment share moving little (Table \ref{tab:ny}; Figure \ref{fig:nyshares}).

The discrete-choice interpretation sharpens the economic distinction. If a revenue source is
viewed as the outcome of an aggregate financing-choice rule, then the clr SCM selects states whose relative reliance on federal, state,
and local finance evolves similarly to New York's; its weights summarize similarity in the
underlying financing incentives or institutional preferences, rather than similarity in each
share considered in isolation. The Gap Elimination Adjustment is then a negative shift in the
relative index for state aid, compensated primarily by local finance. This is an interpretation,
not a structural estimate of household or policymaker preferences: such a claim would require
micro-level choice data and an explicit institutional model.

They nonetheless reach different conclusions. The Euclidean method places $0.54$ of its
weight on Ohio alone, while the Aitchison method spreads weight across fourteen donors led
by Virginia at $0.27$; the two weight vectors differ by $0.418$ in sup norm. New York's state and local shares are both near $46\%$ and dominate a squared-difference objective, so matching them well is close to matching a single series, whereas the log-ratio objective must also reproduce the $7\%$
federal share. The consequence is inferential. New York's post-to-pre ratio ranks second of
fifty units under the Aitchison method, giving an exact permutation $p$-value of $0.040$,
and seventh under the Euclidean method, giving $0.140$. One analyst rejects the null of no
compositional effect at the five percent level and the other does not, from the same data,
the same treated unit, the same donor pool and the same test. Across the raw data and eight
common perturbations, the Aitchison $p$-value is $0.040$ every time while the Euclidean one
ranges from $0.14$ to $0.22$.

The Adjustment was a response to the post-2008 fiscal shortfall, which affected
every state, so part of the movement is common rather than specific to New York; the donor
pool absorbs common movements, exactly so in clr coordinates by Proposition
\ref{prop:cancel}, but the cut's size and formula were New York's own and the timing of
restoration was discretionary. And seven post-treatment years is a short window. I present
this case to show that the divergence documented in Section \ref{sec:why} reaches the
conclusions a researcher would report, not as a definitive evaluation of the policy.

\section{Conclusion}\label{sec:concl}

For compositional outcomes, the synthetic control should be built in clr coordinates. Some Properties separate the methods. Information lives in ratios, so the Euclidean method tilts toward the largest category when selecting donors and inflates high-dispersion categories when averaging them. Common shocks are multiplicative perturbations under the discrete-choice representation, and they cancel
from the Aitchison method at both steps and from neither step of the Euclidean one. The simulations show the advantage for economic interpretation. The implementation cost is only one line.

\end{document}